\documentclass[10pt,journal,cspaper,compsoc]{IEEEtran}

\usepackage{color}
\usepackage{times}
\usepackage{amsmath}
\usepackage{graphicx}
\usepackage{subfig}
\usepackage{epsfig}
\usepackage[named]{algo}
\usepackage{setspace}
\usepackage{multirow}
\usepackage{amssymb}
\usepackage{setspace}
\usepackage{fixltx2e}

\newtheorem{theorem}{Theorem}

\newtheorem{proposition}{Proposition}

\newtheorem{Definition}{Definition}

\newenvironment{proof}[1][Proof]{\begin{trivlist}
\item[\hskip \labelsep {\bfseries #1}]}{\end{trivlist}}

\newcommand{\qed}{\nobreak \ifvmode \relax \else
      \ifdim\lastskip<1.5em \hskip-\lastskip
      \hskip1.5em plus0em minus0.5em \fi \nobreak
      \vrule height0.75em width0.5em depth0.25em\fi}

\begin{document}
%
\title{Algorithms for Efficient Mining of Statistically Significant Attribute Association Information}

\author{Pritam~Chanda,
        Aidong~Zhang,
        and~Murali~Ramanathan,
        }
\IEEEcompsocitemizethanks{\IEEEcompsocthanksitem P. Chanda and A. Zhang are in the Dept. of Computer Science and Engineering, State University of New York, Buffalo, NY\protect\\
E-mails: pchanda@buffalo.edu, azhang@buffalo.edu
\IEEEcompsocthanksitem M. Ramanathan is with Dept. of Pharmaceutical Sciences, State University of New York, Buffalo, NY}
\thanks{}

\IEEEcompsoctitleabstractindextext{%
\begin{abstract}
Knowledge of the association information between the attributes in a data set provides insight into the underlying structure of the data and explains the relationships (independence, synergy, redundancy) between the attributes and class (if present). Complex models learnt computationally from the data are more interpretable to a human analyst when such interdependencies are known. In this paper, we focus on mining two types of association information among the attributes - correlation information and interaction information for both supervised (class attribute present) and unsupervised analysis (class attribute absent). Identifying the statistically significant attribute associations is a computationally challenging task –- the number of possible associations increases exponentially and many associations contain redundant information when a number of correlated attributes are present.  In this paper, we explore efficient data mining methods to discover non-redundant attribute sets that contain significant association information indicating the presence of informative patterns in the data.  
\end{abstract}

}

\maketitle

\IEEEdisplaynotcompsoctitleabstractindextext
\IEEEpeerreviewmaketitle

\section{Introduction}
Many applications in various fields of scientific research, economics, financial and marketing applications produce multi-dimensional data sets in which complicated interdependencies exist between the attributes of data, such as independence, correlation, synergy, and redundancy. Data mining and statistical techniques have been employed to make sense of these data sets, to discover useful patterns and models in the data that aid explaining how the system being represented works. To discover key patterns in the data, it is necessary to find relationships or associations between the attributes in the data that help to explain the interdependencies among the attributes. Exploring attribute association patterns enable deeper insight into the data, are useful for understanding probabilistic models representing the data and possibly allow one to gain practical knowledge from the model(s)computationally learnt using the data.

From an information theoretic perspective, association information between attributes can be broadly categorized into (1) correlation information and (2) interaction information. The correlation information of an attribute set represents the total amount of information shared among the attributes; equivalently, it can be viewed as a general measure of dependency. The interaction information of an attribute set captures the multivariate dependencies between the attributes which is not present in any subset of the given set. These two are related and complements each other in discovering useful patterns and relationships in the data.

\vspace{-2mm}

\section{Background and Significance}
In this paper, we study the problem of mining the above two types of association information that are statistically significant in discrete data for both supervised (i.e. when a class label attribute is present) and unsupervised analysis (no class label is present). Note that the two analysis methods are different because in the first case we need to find sets of attributes that have significant association information with one another, while in the second case we need to find attributes that have significant association information for the class attribute. Finding these types of associations have important implications in many fields of study. For example, in a biological or genetic context, the risk of developing many common and complex diseases such as different forms of diabetes, mental illness, cancer, autoimmune and cardiovascular diseases involves complex interactions between multiple genes and several endogenous and exogenous environmental factors. For many common diseases, individually each gene (or single nucleotide variations on that gene) have weak statistical associations with the disease, however, together they act in concerted fashion (often with several non-genetic factors e.g. gender, age, smoking habits, drinking habits) to control the expression of the disease \cite{Chanda2007,Chanda2008}. The successful detection of such genetic associations can provide the scientific basis for many underlying biological interactions, improves the prospects for uncovering potentially undiscovered genes involved in the disease process and helps to develop preventative and curative measures for particular genetic and non-genetic susceptibilities. Besides genetics, the usefulness of exploring association information is also important in supervised learning problems such as feature selection where the task is to find a subset of the features that improve the accuracy of a classifier. A statistical association between two attributes exists when the joint effect of both in a model is different from that obtained by additively combining the individual effects. Associations among the attributes are specially important for understanding an appropriate probabilistic model representing the data and subsequent feature selection. Discovering associations between the attributes in a data set provides insight into the underlying structure of the data and explains the relationships (independence, synergy, redundancy) between the attributes. Complex models learnt computationally from the data are more interpretable to a human analyst when such interdependencies are known.

\vspace{-2mm}

\section{Related Work}
Mining correlation information in high-dimensional discrete data has attracted much research interest in recent years. Various approaches have been developed, including correlation pattern mining \cite{Lee,Ke,Omiecinski}, feature selection \cite{Fleuret2004,Tesmer2004,Peng2005}, finding correlated item pairs \cite{Xiong}, and others. Mining correlation information is also closely associated with mining frequent patterns in the data. It roots from the association rule mining problem introduced in the Apriori algorithm \cite{Apriori}. Since then much work has been done on frequent pattern mining with itemsets, constrained rule mining, measuring interestingness of association rules mined and so on. Traditionally support and confidence and related measures have been used to assess the usefulness of the rules mined. Correlation pattern mining was achieved with a statistical basis in \cite{Brin} where the authors have used $\chi^{2}$ correlation measure between pairs of attributes. Information theory based metrics like entropy has also been used as a quality measure for sets of attributes (or items) and efficient algorithms have been proposed to mine the maximally informative $k$-itemsets as in \cite{Knobbe}. Algorithms have been proposed to find low-entropy sets as in \cite{Heikinheimo} where they introduced two kinds of low entropy trees and discussed their properties. In the NIFS method \cite{NIFS}, the authors explore the problem of finding non-redundant high order correlations in binary data and propose pruning strategies by investigating the bounds of multi-information which is a  generalization of pair-wise mutual information. Their proposed pruning methods are based on hard thresholds which is difficult to set unless pre-determined using trial and error. Here we derive bounds on correlation information for both supervised and unsupervised analysis, use pruning strategies using bounds on correlation information, however, instead of hard thresholds, we employ the distributional properties of correlation information which improves the power of our method in the presence of noise in the data. Also our bounds are based on entropy inequalities and therefore not restricted to binary data. Using experimental data sets, we further show that our methods can identify attribute sets (we call them \emph{special combinations of interest}) which are not detected in \cite{NIFS} and also mine interaction information among the attribute sets and use a novel fast permutation strategy to evaluate the statistical significance of interaction information of attribute sets.

Compared with correlation information, interaction information is a more parsimonious measure of association. Interaction information between variables and attributes was researched upon in diverse areas like physics, information theory, neuroscience, game theory, law and economics. The concept was first introduced by McGill \cite{McGill1955} as a multivariate generalization of Shannon's mutual information \cite{Shanon1948}. Later, Han \cite{Han1980} gave rigorous formal definitions of the concepts of interaction while properties of positive and negative interactions appeared in \cite{Tsujishita1995}. In physics, Cerf \cite{Cerf1997} analyzed interaction information of three variables in quantum physics, while Matsuda \cite{Matsuda2000} studied properties of interaction information (referred to as higher order mutual information functions) for general complex systems. Bell \cite{Bell2003} defined co-information forming a partially ordered lattice in terms of the entropies and used it for dependent component analysis. More recently, Jakulin \cite{Jakulin2003, Jakulin2004b} studied it extensively from a machine learning perspective and provided methods for visualizing interactions and interpreting the structure in the data.

Correlation measures such as Pearson's correlation, Spearman's rank correlation, Kendall tau correlation and chi square measures are common examples of first order association measures used to evaluate individual attribute dependencies (synergy with class label) or relevance of an attribute in predicting the class label. Associations among attributes have been used for feature selection directly or indirectly in various data mining and machine learning applications, however, most of these consider only first order associations (mutual information) \cite{Nazareth2005, Fleuret2004,Tesmer2004}. Mutual information was also used as a similarity measure for clustering instances \cite{Krasov2005,Zhou2004}. Also, in mining attribute associations it is important to consider the presence of correlated attributes as this results in several associations that contain redundant information regarding the class label. Feature selection methods that explore means to reduce redundancy among the attributes are  studied by some researchers \cite{Hall1998,Peng2005,Pudil2003}. For example in \cite{Peng2005}, the authors devise a minimal-redundancy-maximal-relevance (mRMR) criterion using information theoretic methods to reduce redundancy and select promising features. In CfsSubsetEvaluation \cite{Hall1998} subsets of features that are highly correlated with the class while having low inter-correlation are preferred. Although methods as in \cite{Alves2005} (GRAD) and \cite{Kwak2002} directly or indirectly considers higher order associations, they do not address the problem posed by the presence of large number of correlated variables in the data. Mining highly-correlated association patterns are also explored in \cite{Vipin2002, Vipin2006}. An important difference of our work from others is that we mine higher order association information, consider redundancy among the attributes instead of simple pairwise correlations between attributes as in \cite{Hall1998, Peng2005} and use statistical significance based pruning strategies (unlike  \cite{Vipin2002, Vipin2006,Peng2005}) to improve the efficiency of our search methods for both supervised and unsupervised analyses.

\section{Contributions of the paper}
Mining the significant attribute associations in a high dimensional data set is a computationally challenging task - the number of possible associations increases exponentially because all possible subsets of the attributes need to be considered and most of these associations contain redundant information when a number of correlated attributes are present. Although, in practice, the attribute associations of interest that are meaningful are much fewer in number compared to all possible associations, the high dimensionality of the data sets makes the number of relevant attribute associations very large. Exploring all subsets of attributes for significant association information becomes computationally intractable as the number of attributes increases. In this paper, we study the problem of mining statistically significant correlation information and interaction information in discrete data in both unsupervised and supervised contexts. Our work is based on the concepts developed in \cite{Chanda2010} where we have developed the algorithms for unsupervised analysis only. In this paper, we do the following:

\begin{enumerate}
\item We present the information theoretic metrics representing interaction information (termed K-way interaction information or KWII), correlation information for unsupervised analysis (termed total correlation information or TCI) and correlation information for supervised analysis (termed class associated correlation information or CACI).
\item We demonstrate and prove the relationships between the above association information metrics.
\item We derive the distributional properties of TCI and CACI for evaluating statistical significance or correlation information.
\item We develop a method for fast evaluation of statistical significance of the interaction information (i.e. KWII).
\item For both supervised and unsupervised cases, we propose the concepts of attribute combinations containing highly significant, moderately significant and non-significant correlation information. These are used to formulate \emph{combinations of interest} as highly significant attribute sets that have all subsets with non-significant correlation information, and \emph{special combinations of interest} that can have at most one subset with highly significant correlation information.
\item We present bounds on correlation information (both TCI and CACI) and develop several pruning strategies utilizing these bounds to efficiently prune the search space.
\item Using the bounds and pruning strategies, for unsupervised cases, we develop the algorithms correlation information miner (CIM) and interaction information miner (IIM). We also develop the correlation information miner class associated (CIM\textsubscript{CA} for supervised cases).
\item Using several experimental and a real-life data set, we critically examine the effectiveness and efficiency of our proposed mining algorithms.
\end{enumerate}

\vspace{-3mm}

\section{Association Information Metrics}
In this section, we introduce some basic notations that we shall use throughout the paper. In the rest of the paper, the term combination is also used to refer to a set of attributes. A given data set $D$ is represented as a $m \times n$ matrix of discrete values where each row is a sample and each column is an attribute. Let $\zeta = \{A_1;A_2;...;A_n\}$ be the set of attributes in $D$. We treat $A_i$ as a discrete random variable and $p(a_i)$ represents the probability density function of $A_i$. Also, the words 'combination' and 'set' are used interchangeably in the paper referring to a collection of attributes.
\begin{Definition}
  The uncertainty of a discrete random variable $A_i$ is defined by Shannon's entropy \cite{Shanon1948} as,
  \begin{eqnarray}
  H(A_i) = - \sum_{a \in V_i} p(a_i)log(p(a_i)) \nonumber
  \end{eqnarray}
  \label{def1}
\end{Definition}
\vspace{-3mm}
\begin{Definition}
  The interaction information among the $k$ attributes ($k$-way interaction information or $KWII$) in set $S = \{A_1;A_2;...;A_k\}$, $S \subseteq \zeta$, is
  the multivariate generalizations of Shannon's mutual information. It is defined as the amount of information (synergy or redundancy) that is present in the set of attributes, which is not present in any subset of these attributes \cite{Jakulin2003}. The $KWII$ can be written succinctly as an alternating sum of the entropies of all possible subsets $\tau$ of $S$ using the difference operator notation of Han \cite{Han1980}:
  \begin{eqnarray}
  KWII(S) = -\sum_{\tau \subseteq S} {(-1)}^{|S \backslash \tau|} H(\tau) \nonumber
  \end{eqnarray}
  \label{def2}
\end{Definition}
\vspace{-3mm}
The number of attributes $k$ in a combination is called the order of the combination. KWII quantifies interactions by representing the information that cannot be obtained without observing all $k$ attributes at the same time.

In the bivariate case, the KWII is always nonnegative but in the multivariate case, KWII can be positive or negative(positive values indicate synergy between the attributes, negative values indicate redundancy between attributes, and a value of zero indicates the absence of k-way interactions).

\begin{Definition}
  The Total Correlation Information (TCI) involving attributes in set $S = \{A_1;...;A_k\}$ is defined \cite{Han1980}\cite{Jakulin2004a} as,
  \begin{eqnarray}
  TCI(S) &=& \sum_{i=1}^{k} H(A_i)- H(A_1;...;A_k) \nonumber \\
         &=& \hspace{-3mm} \sum_{a_1,..,a_k}\hspace{-2mm}p(a_1...a_k)log_{2}(\frac{p(a_1...a_k)}{p(a_1)...p(a_k)}) \nonumber
  \end{eqnarray}
  \label{def3}
\end{Definition}
\vspace{-2mm}
The TCI is the total amount of information shared among the attributes in the set. A TCI value that is zero indicates that the attributes are independent and the maximal value of TCI occurs when one attributes is completely redundant with the others. An important property of the TCI is that it is always non-negative and increases monotonically with increasing combination size i.e., $TCI(A_1;\cdots;A_k) \leq TCI(A_1;\cdots;A_k;A_{k+1})$.
Next we examine the correlation metrics in a supervised analysis where a class label attribute is present that specifies the labels of each instance in the data. First note that the TCI can be used to calculate the correlation information by treating the class attribute just as one of the attributes in a combination. However, the correlation information represented by the TCI is not free from unnecessary confounding information that does not involve the class attribute. For example, say we are given data with three predictor attributes $A_1$, $A_2$, and $A_3$ and a class attribute $C$. The value of $TCI(A_1;A_2;A_3;C)$ will represent the overall correlation information among these attributes which contains several components viz. $KWII(A_1;A_2)$, $KWII(A_1;A_3)$, $KWII(A_2;A_3)$ and $KWII(A_1;A_2;A_3)$ which do not contain $C$ and any information related to $C$. We therefore present another metric called the Class Associated Correlation Information (or CACI) which is a non-overlapping sum of interaction information about the class attribute for the predictor attributes $A_1$,...$A_k$ and the class $C$. The CACI is obtained from the measure representing the overall dependency among the predictor attributes and the class attribute by removing the  contributions representing the interdependencies (e.g., correlations) among the predictor attributes not related to the class attribute. Accordingly, the CACI is defined by:

\begin{Definition}
  The Class Associated Correlation Information (CACI) involving attributes in set $S = \{A_1;...;A_k\}$ and class $C$ is defined as,
    \begin{eqnarray}
    &&CACI(S;C) = TCI(S;C) - TCI(S)\nonumber\\
    &=& \sum_{a_1}...\sum_{a_k} \sum_{c} p(a_1,...,a_k,c) log (\frac{p(a_1,...,a_k,c)}{p(c)p(a_1...a_k)})
    \end{eqnarray}
  \label{def4}
\end{Definition}
\vspace{-3mm}
In the above definition, the $TCI(A_1;A_2;...;A_k;C)$ term represents the overall dependency among the all the attributes and the class whereas the $TCI(A_1;A_2;...;A_k)$ term represents the inter-dependencies only among the predictor attributes in the absence of the class attribute.
\vspace{-3mm}

\subsection{Properties of TCI}
\begin{proposition}
The TCI increases monotonically with increased combination size.
\label{Prop1}
\end{proposition}

\begin{proof}
  For $k$ attributes $A_1,A_2,...,A_k$, we have:
  \begin{eqnarray}
   && TCI(A_1;...A_k)-TCI(A_1;...A_{k-1}) = \sum_{i=1}^{k} H(A_i) \nonumber \\
   &-& H(A_1 \cdots A_k) - \sum_{i=1}^{k-1} H(A_i) + H(A_1 \cdots A_{k-1}) \nonumber \\
   &=& H(A_k) - H(A_k|A_1...A_{k-1}) \geq 0
   \label{eq0}
  \end{eqnarray}
  The last inequality follows from the fact that the entropy of $A_k$ decreases when information from $A_1,\cdots A_{k-1}$ is known (the vertical bar represents conditional entropy).
\end{proof}

Here, we state the theorems demonstrating the relationships between the above mentioned two information theoretic metrics \cite{Chanda2010}.
\begin{theorem} The TCI of an attribute set $S$ represents the sum of all KWII between two or more attributes from $S$, i.e.,
  $TCI(S) = \sum_{Z \subseteq S, |Z| \geq 2} KWII(Z)$
  \label{Th2}
\end{theorem}


\vspace{-3mm}

\subsection{Properties of CACI}
\begin{theorem} The CACI of an attribute set $S$ and $C$ represents the sum of all KWII between one or more attributes from $S$ and $C$, i.e.,
  $CACI(S;C) = \sum_{Z \subseteq S, |Z| \geq 1} KWII(Z;C)$
  \label{Th3}
\end{theorem}

\begin{proof}
For the set $S = \{A_1;...;A_k\}$ and class $C$, from definition \ref{def3} we have,
\begin{eqnarray}
&&TCI(S;C)=\sum_{i=1}^{k} H(A_i) + H(C) - H(S;C) \nonumber \\
&=&\sum_{i=1}^{k} H(A_i) - H(S) + H(C) + H(S) - H(S;C) \nonumber \\
&=&TCI(S)+TCI(A_1...A_k;C)
\label{eq3}
\end{eqnarray}
Thus using theorem \ref{Th2},
\begin{eqnarray}
&&TCI(A_1A_2...A_k;C)=TCI(S;C)-TCI(S) \nonumber \\
&=&\sum_{\nu \in \{S;C\},|\nu|\geq 2} KWII(\nu) - \sum_{\omega \in \{S\},|\omega|\geq2} KWII(\omega) \nonumber \\
&=&\sum_{\xi \in \{S\},|\xi|\geq 1} KWII(\xi;C)
\label{eq4}
\end{eqnarray}

The term $TCI(A_1A_2...A_k;C)$ is the TCI between the joint distribution of the $k$ attributes and the class attribute; the $TCI(S)$ = $TCI(A_1;...;A_K)$ term is the TCI among the $k$ attributes and $TCI(S;C)$ = $TCI(A_1;A_2;...;A_K;C)$ is the TCI among the $k$ attributes and the class. The above equation is the sum of all possible interactions involving attributes $A_1,A_2,...,A_k,C$ that contains the class attribute $C$. This is defined as the Class Associated Correlation Information or CACI. Thus,
\begin{eqnarray}
CACI(S;C)&=&TCI(S;C)-TCI(S) \nonumber \\
&=&\sum_{\xi \in \{S\},|\xi|\geq1} KWII(\xi;C)
\label{eq5}
\end{eqnarray}
Because information content of each KWII is non-redundant (or non-overlapping) with every other combination and the CACI can be expressed as a sum of KWII values, the CACI is a non-overlapping sum of information about the class attribute.
\end{proof}

\begin{proposition}
CACI is always greater than or equal to zero and increases monotonically with increased combination size (i.e. $CACI(A_1;...A_k;C) \geq CACI(A_1;...A_{k-1};C)$).
\label{Prop2}
\end{proposition}

\vspace{-3mm}

\section{Problem Formulation}
In this section, we shall develop a problem formulation common to both supervised analysis (i.e. class attribute present) and unsupervised analysis (i.e. class attribute absent) and will use either CACI or TCI. For the ease of presentation, lets denote either CACI or TCI by the term CI (standing for correlation information). Whereever applicable, we shall distinguish between the two by using the actual names (CACI or TCI). First we introduce the concepts of Combinations of Interest (or COI) and Special Combinations of Interest (or SCOI). A COI is an attribute set containing high CI such that its proper subsets have low CI, while a SCOI is similar to the COI but can have exactly one proper subset to have high CI. Our definitions of high and low are based on statistical significance levels which is based on distributional properties explored in section 7. Broadly, our goal is to mine the COI, SCOI and combinations with high KWII that represent attribute sets containing non-redundant association information either with class (for supervised studies) or without class. To develop our mining strategy, we first give some formal definitions.

\vspace{-3mm}

\subsection{Definitions for the Unsupervised Case}
First we present the definitions assuming no class attribute is present. Our definitions use the common statistical concept of $Pvalue$. Given an observed value of a test statistic, $Pvalue$ is defined as probability of obtaining a value more extreme than the given one, under the null distribution of the test statistic. Assume that we know the probability distribution function of the TCI. Let $\alpha_{High}$ and $\alpha_{Low}$ be two given significance levels for determining the statistical significance of an observed value of TCI such that $0 < \alpha_{High} < \alpha_{Low}$. Let $S = \{A_1;\cdots;A_k\} \subseteq \zeta$ be a given set of attributes.

\begin{Definition}
$S$ has statistically \textbf{Highly Significant} correlation information if $Pvalue(TCI(S)) < \alpha_{High}$. We refer to such
a combination of attributes as \textbf{Highly Significant Combination} or \textbf{HSC}.
\label{def5}
\end{Definition}

\begin{Definition}
$S$ has statistically \textbf{Non-Significant} correlation information if $Pvalue(TCI(S)) \geq \alpha_{Low}$. We refer to such
a combination of attributes as \textbf{Non-Significant Combination} or \textbf{NSC}.
\label{def6}
\end{Definition}

\begin{Definition}
$S$ has statistically \textbf{Moderately-Significant} correlation information if $\alpha_{High}$ $\leq Pvalue$ $(TCI(S))$ $<\alpha_{Low}$. Such a combination of attributes is called a \textbf{Moderately-Significant Combination} or \textbf{MSC}.
\label{def7}
\end{Definition}
For example, setting $\alpha_{High} = 10^{-10}$ and $\alpha_{Low} = 10^{-3}$, a $Pavlue$ of $10^{-12}$ will be Highly Significant while that of 0.01 will be Non-Significant.

\begin{Definition}
$S$ is a \textbf{C}ombination \textbf{O}f \textbf{I}nterest (or \textbf{COI}) if it satisfies:-
\begin{enumerate}
\item S is a HSC, and
\item Each proper subset of S is a NSC.
\end{enumerate}
\label{def8}
\end{Definition}

However checking all proper $2^{k-1}$ subsets of $S$ is computationally expensive. Let $S_{k-1} \subset S$ with $k-1$ attributes. From the monotonic increasing property of the TCI (property (3) in definition \ref{def3}), $TCI(S) \geq TCI(S_{k-1})$. Therefore, we make the assumption that if $Pvalue(TCI(S))\geq \alpha_{Low}$, then $Pvalue(TCI(S_{k-1}))$ is also $\geq \alpha_{Low}$ as smaller TCI value usually has lower significance. As a result, we only need to check whether the $k-1$ size subsets of $S$ are NSC.

The definition of COI is based on the fact that if $S$ is a HSC and one or more of its subsets are HSC or MSC, then $S$ has redundancy as it has at least one subset with high correlation information. For example, assume set $S=\{A_1;A_2;A_3;A_4\}$ is a HSC and its subsets $S'=\{A_1;A_2\}$ and $S''=\{A_3;A_4\}$ are also HSC. In this case, mining $S'$ and $S''$ are sufficient to capture all the interacting attributes. However, this is a strict condition that need to be relaxed to capture more information as seen in the next definition.

\begin{Definition}
Let $\Gamma_{k}$ denote the set of all subsets of $S$ with $k-1$ attributes. $S$ is a \textbf{S}pecial \textbf{C}ombination \textbf{O}f \textbf{I}nterest (or \textbf{SCOI}) if it satisfies:-
\begin{enumerate}
\item S is a HSC,
\item Exactly one member (say set $X$) $\in \Gamma_{k}$ is a HSC and all others are NSC, and
\item $\Delta_{TCI}$ = TCI($S$)-TCI($X$) is statistically significant at significance level $\alpha_{High}$.
\end{enumerate}
\label{def9}
\end{Definition}

Let $X = S\backslash\{A_k\}$. Then, it can be easily shown that $\Delta_{TCI}$ = $H(A_k)$ + $H(X)$ - $H(S)$ = $TCI(\vec{X};A_k)$, where $\vec{X}$ represents a new attribute formed by the joint of all attributes in $X$. The motivation behind the definition of SCOI is based on the following example. Assume set $S=\{A_1;A_2;A_3;A_4\}$ is a HSC and only its subset $S'=\{A_1;A_2;A_3\}$ is a $HSC$. If $\Delta_{TCI} = TCI(A_1A_2A_3;A_4)$ is significant, $A_4$ is contributing significantly to the increased correlation information. If we only mine $S$ and not $S'$, we lose important association information contributed by $A_4$ \emph{only in combination with $S$}.\vspace{2mm}

\vspace{-3mm}

\subsection{Definitions for the Supervised Case}
Assume that we know the probability distribution function of the CACI. Let $\alpha_{High}$ and $\alpha_{Low}$ be two given significance levels for determining the statistical significance of an observed value of CACI such that $0 < \alpha_{High} < \alpha_{Low}$. Let $S_c = S_c = S \cup \{C\}$ be a given set of attributes including the class attribute.

\begin{Definition}
$S_c$ has statistically \textbf{Highly Significant} class associated correlation information if $Pvalue$ $(CACI(S_c))$ $<\alpha_{High}$. We refer to such
a combination of attributes as \textbf{Highly Significant Combination Class Associated} or \textbf{HSC\textsubscript{CA}}.
\label{def10}
\end{Definition}

\begin{Definition}
$S_c$ has statistically \textbf{Non-Significant} class associated correlation information if $Pvalue$ $(CACI(S_c))$ $\geq \alpha_{Low}$. We refer to such
a combination of attributes as \textbf{Non-Significant Combination Class Associated} or \textbf{NSC\textsubscript{CA}}.
\label{def11}
\end{Definition}

\begin{Definition}
$S_c$ has statistically \textbf{Moderately-Significant} class associated correlation information if $\alpha_{High}$ $\leq Pvalue$ $(CACI(S_c))$ < $\alpha_{Low}$. Such a combination of attributes is called a \textbf{Moderately-Significant Combination Class Associated} or \textbf{MSC\textsubscript{CA}}.
\label{def12}
\end{Definition}

Again following the definitions we presented for the unsupervised, in presence of $C$, we have,
\begin{Definition}
$S$ is a \textbf{C}ombination \textbf{O}f \textbf{I}nterest class associated (or \textbf{COI\textsubscript{CA}}) if it satisfies:-
\begin{enumerate}
\item $S_c$ is a HSC\textsubscript{CA}, and
\item Each proper subset of $S_c$ is a NSC\textsubscript{CA}.
\end{enumerate}
\label{def13}
\end{Definition}

However checking all proper $2^{k-1}$ subsets of $S_c$ is computationally expensive. Following the same argument as in definition of COI, because CACI also has a monotonic increasing property, we only need to check whether the $k-1$ size subsets of $S_c$ are NSC.

Finally we define the case analogous to SCOI,

\begin{Definition}
Let $\Gamma_{k}$ denote the set of all subsets of $S_c$ with $k-1$ attributes such that each subset contains $C$. $S_c$ is a \textbf{S}pecial \textbf{C}ombination \textbf{O}f \textbf{I}nterest class associated (or \textbf{SCOI\textsubscript{CA}}) if it satisfies:-
\begin{enumerate}
\item $S_c$ is a HSC\textsubscript{CA},
\item Exactly one member (say set $X_c$) $\in \Gamma_{k}$ is a HSC\textsubscript{CA} and all others are NSC\textsubscript{CA}, and
\item $\Delta_{CACI}$ = CACI($S_c$)-CACI($X_c$) is statistically significant at significance level $\alpha_{High}$.
\end{enumerate}
\label{def14}
\end{Definition}

The motivation behind the definition of SCOI\textsubscript{CA} is based on the following example. Assume set $S_c=\{A_1;A_2;A_3;A_4;C\}$ is a HSC\textsubscript{CA} and only its subset ${S'}_c =\{A_1;A_2;A_3;C\}$ is a HSC\textsubscript{CA}. If $\Delta_{CAI}$ is significant, $A_4$ is contributing significantly to the increased correlation information with $C$. If we only mine $S_c$ and not ${S'}_c$, we lose important class related association information contributed by $A_4$ \emph{only in combination with $S_c$}.\vspace{2mm}

\vspace{-3mm}

\subsection{Redundancy Considerations}
Next, we consider correlations among data attributes (e.g. linkage disequilibrium in genetic data) which can result in redundancy (i.e. presence of overlapping information) among the attribute combinations. First we present the case for unsupervised analysis. Using the property that KWII is negative in presence of redundancy, we have,

\begin{Definition}
Two attributes $A_i$ and $A_j$ are redundant if $Red(A_i;A_j) = \frac{KWII(A_i;A_j;A_j)}{\min\{H(A_i),H(A_j)\}} \leq -\Delta$, where $0 \leq \Delta \leq 1$ is a user specified threshold.
\label{def15}
\end{Definition}

The definition is based on the fact that if $A_i$ and $A_j$ have high redundancy, they are in fact interacting, i.e, $A_i$ explains $A_j$ very well. Also $A_j$ completely explains itself ($A_j$) causing the expression $KWII(A_i;A_j;A_j)$ to have redundant information. The denominator is used to normalize the KWII and is based on the easy to prove fact that $KWII(A_i;A_j;A_k) \leq min\{H(A_i),H(A_j),H(A_k)\}$.

In presence of a class attribute $C$, we have,
\begin{Definition}
Two attributes $A_i$ and $A_j$ are redundant in the context of $C$ if $Red(A_i;A_j) = \frac{KWII(A_i;A_j;C)}{H(C)} \leq -\Delta_{CA}$, where $0 \leq \Delta_{CA} \leq 1$ is a user specified threshold in the presence of a class variable.
\label{def16}
\end{Definition}
In the above definition, if the variables $A_i$ and $A_j$ are redundant, they have similar information about $C$, as a result, the $KWII(A_i;A_j;C)$ will have redundant information making it negative.

\vspace{-3mm}

\subsection{Mining Strategy}
Compared with the TCI or CACI, the KWII is a more valuable information metric because it is a parsimonious measure of association for the attribute combination of interest alone and does not contain contributions from lower-order combinations \cite{Jakulin2003}. However, KWII alone cannot be used to device an efficient mining algorithm because it takes on both positive and negative values. Only all individual and joint entropies are needed for a TCI or CACI calculation, making it computationally far more tractable than the KWII. Both the TCI and CACI are always non-negative and increases monotonically with increased combination size making it potentially suitable for our mining algorithm. In the unsupervised case, from theorem \ref{Th2}, the TCI represents the cumulative synergy present in all subset combinations of the attribute set  $\{A_1;A2;\cdots;A_k\}$. Our goal is therefore to use the TCI in our mining algorithm to identify the regions in the combinatorial space (the COI and the SCOI) that contain potentially high correlation information (and therefore high interaction information) and then compute the KWII for the reduced combinatorial space. As a result, we shall concomitantly mine attribute sets containing useful correlation information (i.e. TCI) and interaction information (i.e. KWII). Similarly, in presence of a class variable, we shall use the CACI to identify regions in the combinatorial space containing high class associated correlation and interaction information

Given a maximum order of combinations to explore ($K$) and a pair of significance levels ($\alpha_{High}$, $\alpha_{Low}$), our strategy of mining combinations with significant TCI (or CACI) and KWII broadly consists of two steps :-

\begin{enumerate}
\item Mine all combinations that are COI and SCOI (or COI\textsubscript{CA} and SCOI\textsubscript{CA}), and
\item If $\nu$ is the set of attributes present in combinations mined in step 1, compute KWII($\tau$) of all subsets $\tau \subseteq \nu$, s.t. $\tau \leq K$ (or, in presence of class attribute $C$, if $\nu$ is the set of predictor attributes present in combinations mined in step 1, compute KWII($\tau$;C) of all subsets $\tau \subseteq \nu$, s.t. $\tau \leq K$.
\end{enumerate}

In step 1, we explore the search space in a breadth-first manner that results in a set enumeration tree as shown in Figure \ref{Fig1}. When mining for COI and SCOI (or COI\textsubscript{CA} and SCOI\textsubscript{CA}), computing the TCI (or CACI) of every attribute set is time consuming, therefore, in the next section we shall develop upper and lower bounds of TCI (or CACI) based on that of its parent/ancestor/sibling nodes in the search space. We further develop pruning strategies using definitions of COI, SCOI (or COI\textsubscript{CA} and SCOI\textsubscript{CA}) and redundancy (definitions \ref{def5}-\ref{def16}).

\begin{figure}[h]
\centering \epsfig{file=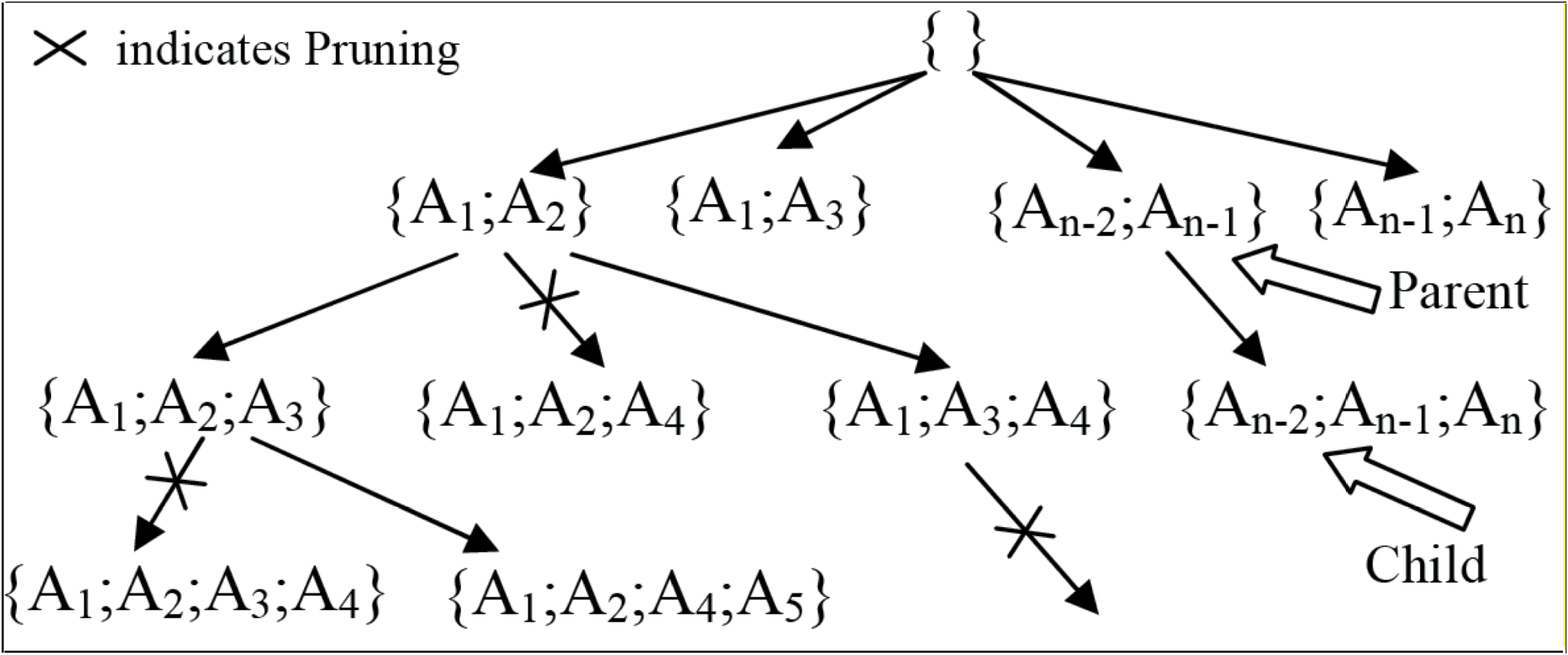,height=1.2in, width=3.0in}
\caption{\label{Fig1}Sample tree enumeration of BFS for Unsupervised Mining.}
\end{figure}

\section{Correlation Information Bounds}
In this section, we present results on upper and lower bounds on TCI and CACI. The $Palue$ computation on these bounds shall be used to speed up our mining strategy.

\vspace{-3mm}

\subsection{Bounds on TCI}
In obtaining the upper and lower bounds, we shall assume TCI computations on the attribute set $S = \{A_1;\cdots;A_k\} \subseteq \zeta$ unless otherwise stated.

\begin{theorem}
  \begin{eqnarray}
  TCI(S) &\geq&  \sum_{i=1}^{k} H(A_i) - \frac{1}{2}[H(S \backslash \{A_1\}) \nonumber \\
        &+& H(S \backslash \{A_2\})+H(A_1;A_2)]\nonumber
  \end{eqnarray}
  \label{Th4}
  \vspace{-3mm}
\end{theorem}
The above theorem computes a lower bound on $TCI(S)$ using entropy from the ancestor nodes. We first use it recursively in computing the upper bound of $H(S)$ in a greedy fashion - first obtain its two-attribute subset (say $(A_i;A_j)$) with maximum pair-wise entropy and then recursively compute upper bounds of the entropies $H(S \backslash \{A_i\}$ and $H(S \backslash \{A_j\})$. The upper bound on $H(S)$ is then used to compute the lower bound of $TCI(S)$.

\begin{theorem}
  \begin{eqnarray}
  TCI(S) \leq TCI(S \backslash \{A_t\}) + \min \{H(S \backslash \{A_t\},H(A_t))\} \nonumber
  \end{eqnarray}
  \label{Th5}
  \vspace{-3mm}
\end{theorem}
The theorem computes a upper bound on $TCI(S)$ using TCI and entropy of its parent node $\{S \backslash \{A_t\}\}$ and $H(A_t)$. The next two theorems are used to compute the upper and lower bounds of the node $\{S;A_j\}$ using entropy of its sibling $\{S;A_i\}$, entropies of individual attributes and conditional entropies. Note that each conditional entropy of form $H(A_i|A_j)$ is given by $H(A_i;A_j)-H(A_j)$.

\begin{theorem}
  \begin{eqnarray}
  TCI(S;A_j) \geq  \sum_{t=1}^{k} H(A_t) + H(A_j) - H(S;A_i) \nonumber \\
  - \min_{t=1}^{k}\{H(A_j|A_t)\}\nonumber
  \end{eqnarray}
  \vspace{-3mm}
  \label{Th6}
\end{theorem}
\begin{theorem}
  \begin{eqnarray}
  TCI(S;A_j) \leq  \sum_{t=1}^{k} H(A_t) + H(A_j) - H(S;A_i) + \Lambda \nonumber \\
  \mbox{where,} \hspace{10mm} \Lambda = \min\{H(A_i|A_j),\min_{t=1}^{k}\{H(A_j|A_t)\}\}\nonumber
  \end{eqnarray}
  \label{Th7}
  \vspace{-3mm}
\end{theorem}

\vspace{-5mm}

\subsection{Bounds on CACI}
In obtaining the upper and lower bounds, we shall assume CACI computations on the attribute set $S = \{A_1;\cdots;A_k\} \subseteq \zeta$ and class variable $C$ unless otherwise stated. A lower bound on CACI is given by the following theorem.
\begin{theorem}
  \begin{equation}
  CACI(S;C) \geq H(C) - \min_{i=1}^{k}H(C|A_i)
  \end{equation}
  \label{Th8}
\end{theorem}
\begin{proof}
  We have $CACI(S;C)$ = $H(A_1...A_k)+H(C)$ - $H(A_1...A_kC)$ = $H(C)$ - $H(C|A_1...A_k)$. Now the result follows from the fact that $H(C|A_1...A_k) \leq H(C|A_i)$ $\forall i=1...k$.
\end{proof}

The following theorem gives us an upper bound on CACI.
\begin{theorem}
  \begin{eqnarray}
  && CACI(S;C) \leq  \nonumber \\
  && \min\{ \frac{1}{2}[H(S \backslash \{A_1\}) + H(S \backslash \{A_2\})+H(A_1;A_2)] , C\} \nonumber
  \end{eqnarray}
  \label{Th9}
\end{theorem}
\begin{proof}
  We have $CACI(S;C) = H(C) - H(C|S)$ so that $CACI(S;C) \leq H(C)$. Again, $CACI(S;C) = H(S) - H(S|C)$ so that $CACI(S;C) \leq H(S)$. Thus clearly,
  $CACI(S;C) \leq \min\{H(C),H(S)\}$. But $H(S) \leq $ $\frac{1}{2}[H(S \backslash \{A_1\})$ + $H(S \backslash \{A_2\})$ + $H(A_1;A_2)]$ (Theorem 6.1 eq 6.3 in \cite{Chanda2010}, the result follows from that.
\end{proof}

\section{Statistical Significance of Correlation and Interaction Information}
\subsection{Probability Distribution of TCI}
In this section, we state results on the probability distribution of the TCI using a Taylor series based approximation to the TCI \cite{Chanda2010}. This shall be used to evaluate the significance of the correlation information of an attribute set.
\begin{theorem}
  The distribution of $\widehat{TCI}(A_1;\cdots;A_k)$ can be approximated by a gamma distribution with scale parameter = $1/(N \hspace{1mm} ln(2))$ and shape parameter = ${df}_{TCI}/2$.
  \label{Th11}
\end{theorem}
Using theorem \ref{Th11}, the $Pvalue$ of an observed TCI value $t$ is given by $Prob(TCI > t)$.

Next we derive the probability distribution of the $CACI$ random variable.

\vspace{-3mm}

\subsection{Probability Distribution of CACI}
We derive the probability distribution of the CACI. The proof is very similar to the one for TCI.
\begin{theorem}
  Let $S = \{A_1;...;A_k\}$ denote a set of variables and $C$ be a class variable. Let $\vec{A}$ represent a new variable formed by the joint of all attributes in $S$. Then the CACI can be approximated as,
  \begin{eqnarray}
  CACI(S;C) \hspace{-1mm}\approx \hspace{-1mm} \frac{1}{2ln(2)} \sum_{a_1,..,a_k,c}\hspace{-2mm}\frac{(p(\vec{a},c)-p(\vec{a})p(c))^2}{p(\vec{a})p(c)}\nonumber
  \end{eqnarray}
  \label{Th12}
\end{theorem}

\begin{proof}
  Let $p(\vec{a}) = \psi_1$ and $p(\vec{a})p(c) = \psi_2$. \\
  Let $f(\psi_1) = p(\vec{a})log_{2}(\frac{p(\vec{a})}{p(\vec{a})p(c)})$ = $\psi_1 log_2 (\frac{\psi_1}{\psi_2})$ = $\frac{\psi_1}{ln(2)} ln (\frac{\psi_1}{\psi_2})$. \\
  Using Taylor's expansion of $f(\psi_1)$ about $\psi_1 = \psi_2$, we have,\\
  $f(\psi_1) = f(\psi_2) + f'(\psi_2)\frac{(\psi_1-\psi_2)}{1!} + f''(\psi_2)\frac{(\psi_1-\psi_2)^2}{2!} + ...$ \\
  Here, $f'(\psi_1)$ = $\frac{ln(\psi_1)+1-ln(\psi_2)}{ln(2)}$ and $f''(\psi_1)=\frac{1}{\psi_1 ln(2)}$. Therefore, $f(\psi_1) = \frac{\psi_1 - \psi_2}{ln(2)} + \frac{1}{2ln(2)\psi_2}(\psi_1-\psi_2)^2 + ...$. \\
  Ignoring higher order terms in the Taylor's expansion,

  $CACI(S;C)$ \hspace{-2mm} = \hspace{-2mm} $\sum_{a_1,..,a_k}f(\psi_1)$ = \hspace{-2mm} $\sum_{a_1,..,a_k}$ $\frac{\psi_1 - \psi_2}{ln(2)}$ + $\frac{1}{2ln(2)\psi_2}(\psi_1 - \psi_2)^2$ + $...$ $\approx$ \hspace{-2mm} $\sum_{a_1,..,a_k}\frac{\psi_1}{ln(2)}$ - $\sum_{a_1,..,a_k} \frac{\psi_2}{ln(2)}$ + $\sum_{a_1,..,a_k} \frac{(\psi_1 - \psi_2)^2}{2ln(2)\psi_2}$.
  The first two summations sum to $1/ln(2)$ resulting in theorem \ref{Th12}.
\end{proof}

Again, the expression of $CACI$ is related to the two-dimensional statistical $\chi^2$ test \cite{Sheskin} defined as,
\begin{equation}
\chi^2 = \sum_{i_1,i_2} \frac{(O_{i_1,i_2}-E_{i_1,i_2})^2}{E_{i_1,i_2}}
\end{equation}
where the summation is over the cells of the $2$-dimensional contingency table, $O_{i_1,i_2}$ denotes the observed cell count and $E_{i_1,i_2}$ denotes the expected cell count for cell $i_1,i_2$. The degrees of freedom present in a $2-$ dimensional contingency table is ${df}_{CACI}$ = $R_1 R_2 - R_1 - R_2 + 1$ \cite{Sheskin}. Here $R_1$ denotes the count of distinct values that variable $\vec{A}$ can take, while $R_2$ denotes the count of distinct values that variable $C$ can take. Equating the observed and expected cell counts to the relative frequencies and the cell probabilities, it can be easily observed that,
\begin{equation}
\chi^2 = 2\hspace{1mm}N\hspace{1mm}ln(2)\hspace{1mm} \widehat{CACI}(A_1;\cdots;A_k;C)
\label{eq13}
\end{equation}
where $N$ denotes the total number of samples in the data (i.e. sum of cell counts in all cells of the $2$-dimensional contingency table). $\widehat{CACI}$ represents the approximation to the $CACI$ metric. Using theorem \ref{Th11} and equation \ref{eq13} it can be easily proved that,

\begin{theorem}
  The distribution of $\widehat{CACI}(A_1;\cdots;A_k;C)$ can be approximated by a gamma distribution with scale parameter = $1/(N \hspace{1mm} ln(2))$ and shape parameter = ${df}_{CACI}/2$.
  \label{Th13}
\end{theorem}
Using theorem \ref{Th13}, the $Pvalue$ of an observed CACI value $t$ is calculated as $Prob(CACI > t)$.

\vspace{-3mm}

\subsection{Probability Distribution of $\Delta$ CACI}
\begin{theorem}
Let $S_c$ = $\{A_1;A_2;\cdots;A_k;C\}$, let $X_c$ = $S \backslash \{A_k\}$ = $\{A_1;A_2;\cdots;A_{k-1};C\}$. Then $\Delta_{CACI}$ = $CACI(S_c)$ - $CACI(X_c)$. Let $|A_k|$ represent the number of states of the attribute $A_k$. Let $\vec{A}$ represent a new variable formed by the joint of all attributes in $\{A_1;A_2;\cdots;A_{k-1}\}$ and  $\vec{a}$ represents its realizations. The distribution of $\Delta_{CACI}$ can be approximated by a gamma distribution with scale parameter = $1/(N \hspace{1mm} ln(2))$ and shape parameter = $|\vec{A}|{df}_{CACI}/2$.
  \label{Th14}
\end{theorem}

\begin{proof}
First note that $\Delta_{CACI}$ can be written as,
\begin{eqnarray}
\Delta_{CACI} &=& \sum_{a_k,\vec{a},c} p(a_k,\vec{a},c) log (\frac{p(a_k,\vec{a},c)p(\vec{a})}{p(c,\vec{a})p(a_k,\vec{a})}) \nonumber \\
              &=& \sum_{a_k,\vec{a},c} p(a_k,\vec{a},c) log (\frac{p(a_k,c|\vec{a})}{p(c|\vec{a})p(a_k|\vec{a})}) \nonumber \\
              &=& \sum_{\vec{a}} p(\vec{a}) \sum_{a_k,c} p(a_k,c|\vec{a}) log (\frac{p(a_k,c|\vec{a})}{p(c|\vec{a})p(a_k|\vec{a})}) \nonumber \\
              &=& \sum_{\vec{a}} p(\vec{a}) \Psi_{\vec{a}}(A_k;C) \mbox{ (Let)}
\end{eqnarray}

Assuming the random variables $A_K$ and $C$ are independent given $\vec{A}$, the expression $\Psi_{\vec{a}}(A_k;C)$ = 0 because $p(a_k,c|\vec{A}=\vec{a})$ = $p(a_k)P(c)$, $p(a_k|\vec{A}=\vec{a})$ = $p(a_k)$ and $p(c|\vec{A}$ = $p(c)$. Therefore we can assume $\Psi_{\vec{a}}(A_k;C)$ to be an independent random variable given each value $\vec{a}$ of the random variable $\vec{A}$. Now note that given a specific value $\vec{a}$ of the random variable $\vec{A}$, we have,
\begin{eqnarray}
\Psi_{\vec{a}}(A_k;C) &=& \sum_{a_k,c} p_{\vec{a}}(a_k,c) log (\frac{p_{\vec{a}}(a_k,c)}{p_{\vec{a}}(c)p(a_k)}) \nonumber \\
                      &=& CACI_{\vec{a}}(A_k;C)
\end{eqnarray}
In the above equation, $p_{\vec{a}}$ represents the probabilities calculated only using the data samples with $\vec{A} = \vec{a}$ and $CACI_{\vec{a}}$ represents the corresponding CACI. Therefore $\Psi_{\vec{a}}(A_k;C)$ is gamma distributed with scale $1/N_{\vec{a}} ln(2)$, shape ${df}_{CACI}/2$ and moment generating function $M_{\Psi_{\vec{a}}}(t) = (1-\frac{t}{N_{\vec{a}} ln(2)})^{-\frac{{df}_{CACI}}{2}}$. $N_{\vec{a}}$ represents the data samples with $\vec{A} = \vec{a}$. As a result, $\Delta_{CACI}$ can be considered as a weighted sum of independent gamma random variates. Therefore the moment generating function of $\Delta_{CACI}$ is given by

\begin{eqnarray}
M_{\Delta_{CACI}}(t) &=& \prod_{\vec{a}} M_{\Psi_{\vec{a}}}(p(\vec{A}=\vec{a})t) \nonumber \\
                     &=& \prod_{\vec{a}} (1-\frac{p(\vec{A}=\vec{a})t}{N_{\vec{a}} ln(2)})^{-\frac{{df}_{CACI}}{2}}
\end{eqnarray}
But $p(\vec{A}=\vec{a})$ = $N_{\vec{a}}/N$ so that $p(\vec{A}=\vec{a})/N_{\vec{a}}$ = $1/N$. Therefore,

\begin{eqnarray}
M_{\Delta_{CACI}}(t) &=& \prod_{\vec{a}} (1-\frac{t}{N} ln(2))^{-\frac{{df}_{CACI}}{2}} \nonumber \\
                     &=& (1-\frac{t}{N} ln(2))^{-\frac{|\vec{A}|{df}_{CACI}}{2}}
\end{eqnarray}
which is the moment generating function of the gamma distribution with scale parameter = $1/(N \hspace{1mm} ln(2))$ and shape parameter = $|\vec{A}|{df}_{CACI}/2$.
\end{proof}

Using theorem \ref{Th13}, the $Pvalue$ of an observed CACI value $t$ is calculated as $Prob(CACI > t)$.

\vspace{-3mm}

\subsection{Significance of Interaction Information}
Determining a closed form expression of the KWII is difficult as KWII involves alternating sums of the entropies of all possible subsets unlike TCI and CACI. We therefore resort to a permutation strategy to calculate the significance (i.e. the $Pvalue$) of KWII of a set of attributes. The strategy is slightly different for unsupervised and supervised analysis. First consider the case of unsupervised analysis. Assume that we want to calculate the significance of $t$ = $KWII(A_1;A_2;\cdots;A_k)$ = $KWII(\omega)$. Let $X$ be the attribute from the set $\{A_1;A_2;\cdots;A_k\}$ with the minimum number of states. For supervised analysis to calculate the significance of $t$ = $KWII(A_1;A_2;\cdots;A_k;C)$ = $KWII(\omega)$ ($C$ being the class attribute), let $X = C$. Our permutation procedure will shuffle the states of the attribute $X$. Then the following algorithm calculates the $Pvalue$ of value $t$:\\

\noindent \textbf{PERMUTATION}($\omega,t$)\\
1.\hspace{2mm}$KWII_{actual} \leftarrow t$;\\
2.\hspace{2mm}Generate $NPERM$ permutations of the data by randomly shuffling the states of the attribute $X$;\\
3.\hspace{1mm}Calculate permuted $KWII^{i}(\omega)$ for each permuted data $i$;\\
4.\hspace{2mm}$Pvalue \leftarrow $ fraction of all $KWII^{i}(\omega) \geq KWII_{actual}$; \\
5.\hspace{2mm}\textbf{return} $Pvalue$;\\
\textbf{end}\\

\noindent \textbf{Fast Permutations} The permutation procedure described in the algorithm, if implemented naively, can be very time consuming because for a given combination the KWII needs to be computed across the entire data samples repeatedly for each shuffle of the states of the attribute $X$. However for discrete, we can implement the permutation in a faster manner. The key observation is that for a given combination attributes (and possibly class $C$), the sufficient statistics for computing the KWII (the empirical counts for each state for different subsets of the attributes) are present in the corresponding contingency table in which the rows represent the states of the attributes (except $X$) while the columns represent the states of $X$. As a result, a shuffle of the states of $X$ corresponds to a change in counts in the cell of the contingency table such that the row sums and column sums are unchanged. Note that, we only need to scan the data once to build the contingency table for each combination which is required anyway for computing the original KWII. Once we have the contingency table for a particular combination, we can shuffle the counts in the contingency table in the above manner and compute the KWII for each shuffled table to compute the permuted KWII values. Assume a combination $C$ with $k$ variables has $b$ states, then the contingency table $T$ will have $b$ cells. Creating $T$ has $O(m \times b)$ complexity where $m$ is the sample size of the data. Then KWII($C$) requires $O(m \times b$ + $2^k \times b)$ = $O(m \times b)$ computations (for $m >> 2^k$) as entropies of all subsets of $k$ variables are computable by marginalizing $T$. Thus the first KWII computation involves $O(m \times b)$ computations because $T$ is constructed. For each permutation, we shuffle the counts in $T$ using an efficient algorithm presented in \cite{Patefield1981} which consumes approximately $O(b)$ computations, so that for $NPERM$ permutations, time complexity if only $NPERM \times O(b)$. Also the KWII constitutes the output from the IIM algorithm and we anticipate very few interactions to be present in the data, so permutations need to be performed on few attribute combinations.

\vspace{-4mm}

\section{Algorithm}
In this section we describe our mining algorithms in details. The algorithms developed will be for unsupervised analysis. The same algorithms with some modifications can be applied for supervised analysis, therefore, the modifications for class associated analysis will be described in context. Our mining algorithm consists of two stages -(1)Correlation Information Miner (or CIM) followed by (2) Interaction Information Miner (or IIM). The CIM explores the combinatorial space of attribute sets using a breadth-first search (BFS) enumerating a BFS tree where each node represents an attribute set $\{A_i;A_j;\cdots;A_k\}$ ($i \le j \le \cdots \le k$)(or, $\{A_i;A_j;\cdots;A_k;C\}$ ($i \le j \le \cdots \le k$) when $C$ is present). Next we describe pruning strategies using the concept of redundancy and bounds on TCI (or CACI) introduced before.

\vspace{-3mm}

\subsection{Redundancy based pruning}
\label{Red}
This pruning strategy is applied to the given data set $D$ before starting the BFS strategy using the redundancy definition \ref{def9}. The goal is to remove redundant attributes thereby reducing the size of the combinatorial space of attribute associations. It consists of (I) For each attribute $A_i \in \zeta$, compute $Red(A_i;A_j)$ with every other attribute $A_j \in \zeta$. If $Red(A_i;A_j)\leq-\Delta$, store $A_j$ in a list associated with $A_i$. This step will create a list of attributes redundant with each $A_i$ denoted as $Cover(A_i)$ (which includes $A_i$). An attribute $A_j \in Cover(A_i)$ is said to be covered by $A_i$. E.g. if $A_1$ is redundant with $A_2$,$A_5$ and $A_8$, Cover($A_1$) = $\{A_1;A_2;A_5;A_8\}$. (II) Create a smaller data set $D'$ by greedily selecting attribute $A_i$ with highest cardinality $|Cover(A_i)|$ (i.e, covering the maximum number of other attributes) until all attributes $\in \zeta$ are covered. This smaller data set will be used as input for the algorithm described below. The computation of $Red(A_i;A_j)$ will use either definitions \ref{def15} or \ref{def16} depending on whether $C$ is absent or present in the analysis.

\vspace{-3mm}

\subsection{Sample Size based pruning}
Given attribute set $S=\{A_1;\cdots;A_k\}$ and sample size $N$, TCI($S$) and KWII($S$) are based on empirically estimated probabilities distributions of the attributes and their combinations from set $S$. Let the cardinality of the set of attribute values of $S$ be $V$. The calculated TCI and KWII are often poor estimates when $N/V<5$ \cite{Agresti}. Therefore, we prune node $S$ when $N/V<5$ to reduce the chances of discovering false positive associations. For example, to evaluate TCI of $\{A_1;A_2;A_3\}$ where attribute takes 3-values, there should be least $3^3 \times 5 = 135$ instances. Similarly, for supervised analysis, with $S=\{A_1;\cdots;A_k;C\}$, we prune node $S$ when $N/V<5$, where $V$ is the cardinality of the values of the attributes in set $S$.

\vspace{-3mm}

\subsection{Bound based pruning}
First we describe the pruning for unsupervised analysis that uses the TCI.

\vspace{-3mm}

\subsubsection{TCI Bound based pruning}
For each node $S$ in the search space, we calculate its upper and lower bounds before actual $TCI(S)$. Let $L(S)$ be the maximum of the lower bounds, $U(S)$ be the minimum of the upper bounds and $T(S)$ be the true TCI for node $S$. Let $P(v)$ be the $Pvalue$ for any value $v$. Note that as $L(S) \leq T(S) \leq U(S)$, we have $P(L(S)) \geq P(T(S)) \geq P(U(S))$. Assume that we have determined if $S$ is a HSC/MSC/NSC. We shall employ the procedures \textbf{Handle HSC} and \textbf{Handle MSC/NSC} described below to handle each case. In the following, in each iteration, $NextLevel$ is a queue that collects nodes to be explored in the next iteration of BFS and $\Theta$ is a set of COI/SCOI output by CIM.

(1) \textbf{Handle HSC} : Assume that the parent of the given node $S$ is not a COI/SCOI. Using property 2 in definition \ref{def7}, if $S$ is a COI, store $S$ in $\Theta$ and add node $S$ to $NextLevel$; otherwise, prune subtree rooted at $S$ as at least one subset of $S$ has redundant correlation information. If the parent of $S$ is a COI/SCOI, using property 2 and 3 in definition \ref{def8}, if $S$ is a SCOI, store $S$ in $\Theta$ and add node $S$ to $NextLevel$; otherwise, prune subtree rooted at $S$ as the new attribute present in $S$ (and not in its parent) does not significantly increase the correlation information. \vspace{1mm}\\
(2) \textbf{Handle MSC/NSC} : If node $S$ is a MSC, $S$ and any superset of it cannot be a COI/SCOI. So simply prune the subtree rooted at $S$. If it is a NSC, add $S$ to $NextLevel$ to continue the search process.

Based on the TCI bounds, we have the following cases:-

\begin{enumerate}
\item $P(U(S))\leq P(L(S))<\alpha_{High}$ : $S$ is a $HSC$. Use \textbf{Handle HSC} to handle it.

\item $P(U(S)) < \alpha_{High} \leq P(L(S))<\alpha_{Low}$ :
\item $P(U(S)) < \alpha_{High}, \alpha_{Low} \leq P(L(S))$ :
\item $\alpha_{High} \leq P(U(S))<\alpha_{Low} \leq P(L(S))$ : Compute the TCI T($S$). If $Pvalue(T(S)) < \alpha_{High}$, $S$ is a HSC, use \textbf{Handle HSC} to handle it. Otherwise use \textbf{Handle MSC/NSC}.

\item $\alpha_{Low} \leq P(U(S))\leq P(L(S))$ : $S$ is a NSC, use \textbf{Handle MSC/NSC}.
\item $\alpha_{High} \leq P(U(S)) \leq P(L(S))<\alpha_{Low}$ : $S$ is a MSC, use \textbf{Handle MSC/NSC}.
\end{enumerate}
Note that actual TCI computations are required only in cases 2,3 and 4 thereby improving computational efficiency. Next we describe the CIM algorithm.

\vspace{-3mm}

\subsubsection{CACI bound based pruning}
\label{CACI_pruning}
The pruning strategy for supervised analysis is very much similar to the above case. Therefore, we do not describe it in details, rather we highlight the following modifications in the above strategy to do the pruning when $C$ is present in the analysis.
\begin{enumerate}
\item Now each node in the search space $S$ represents the set $\{A_1;\cdots;A_k;C\}$ and we calculate the upper and lower bounds before the actual $CACI(S)$.
\item We now have procedures \textbf{Handle HSC\textsubscript{CA}} and \textbf{Handle MSC\textsubscript{CA}/NSC\textsubscript{CA}}. \textbf{Handle HSC\textsubscript{CA}} operate in exactly the same fashion as \textbf{Handle HSC} with the difference that the definitions of \textbf{COI\textsubscript{CA}} and \textbf{SCOI\textsubscript{CA}} are now used. Similarly for \textbf{Handle MSC\textsubscript{CA}/NSC\textsubscript{CA}}.
\item The six bound based cases mentioned above are also applicable with \textbf{Handle HSC\textsubscript{CA}} and \textbf{Handle MSC\textsubscript{CA}/NSC\textsubscript{CA}} usage.
\item Another modification that should be made is for case 5. When the upper bound $U(S)$ is non-significant at $\alpha_{Low}$ and $U(S)$ has reached the maximum value $H(C)$, the CACI of $S$ or its children can never be significant at $\alpha_{Low}$. Therefore $S$ and all its children will be NSC\textsubscript{CA}, so we can safely prune the subtree rooted at $S$.
\end{enumerate}

\vspace{-3mm}

\subsection{The Algorithms}
As before, we first describe the algorithm for unsupervised analysis (CIM and IIM algorithms). Then we describe the changes to be made in CIM to get CIM\textsubscript{CA} algorithm for supervised analysis.

\vspace{-3mm}

\subsubsection{The CIM Algorithm}
We describe the algorithm for unsupervised analysis, the modifications to CIM for supervised analysis will be described separately. We assume that CIM uses the data obtained after redundancy removal (section \ref{Red}) for all computations of correlation information and the upper and lower bounds. The inputs are the significance levels $\alpha_H$ for $\alpha_{High}$ and $\alpha_L$ for $\alpha_{Low}$. Lines 2-8 computes the TCI for every pair of attributes and stores it in $NextLevel$ only if the node is a HSC or a NSC. The HSC are collected in $\Theta$ to be output. Lines 9-33 explores the combinatorial search space in a breadth-first fashion wherein each node is evaluated to be a HSC/MSC/NSC and either the subtree rooted at the node is pruned or the search process is continued depending upon the TCI bound based conditions 1-6 outlined above. The sample size based pruning takes place in line 14.

\begin{algorithm}{CIM}[\alpha_H,\alpha_L]
{
 \qinput
  {
   $\alpha_H$,$\alpha_L$
  }
  \qoutput
  {
    $\Theta$(set of COI and SCOI)
  }
}
$NextLevel \qlet \phi$;$\Theta \qlet \phi$;\\
\textbf{for} attribute pair $S=\{A_i;A_j\}$ \textbf{do}\\
\hspace{2mm}\textbf{if} ($P(TCI(S))<\alpha_{H}$)\\
\hspace{4mm} Add $S$ to $NextLevel$,$\Theta$;\\
\hspace{2mm}\textbf{elseif} ($P(TCI(S))\geq \alpha_{L}$)\\
\hspace{4mm} Add $S$ to $NextLevel$;\\
\hspace{2mm}\textbf{endif}\\
\textbf{endfor}\\
\textbf{while} $NextLevel\neq$ empty \textbf{do}\\
\hspace{2.5mm}  $CurrLevel \qlet NextLevel$;\\
\hspace{2.5mm}  $NextLevel \qlet \phi$;\\
\hspace{2.7mm}\textbf{for} each $P \in CurrLevel$ \textbf{do}\\
\hspace{3.5mm}   \textbf{for} each child $S$ of $P$ \textbf{do}\\
\hspace{4.5mm}   \textbf{if} not enough samples, goto line 31;\\
\hspace{4.5mm}   Calculate $U(S),L(S),P(U(S)),P(L(S));$\\
\hspace{4.5mm}   \textbf{if}$(P(L(S))< \alpha_H)$ \textbf{do}\\
\hspace{6.0mm}         Handle HSC to update $NextLevel$,$\Theta$;\\
\hspace{4.5mm} \textbf{elseif}($P(U(S)) < \alpha_{H} \leq P(L(S))<\alpha_{L}$) or \\
\hspace{14.5mm} ($P(U(S)) < \alpha_{H}, \alpha_{L} \leq P(L(S))$) or \\
\hspace{10.5mm} ($\alpha_{H} \leq P(U(S))<\alpha_{L} \leq P(L(S))$)) \\
\hspace{7mm} $T \qlet TCI(S)$;\\
\hspace{7mm} \textbf{if}$(Pvalue(T)<\alpha_{H})$\\
\hspace{10mm}    Handle HSC to update $NextLevel$,$\Theta$;\\
\hspace{7mm} \textbf{else}\\
\hspace{9mm}     Handle MSC/NSC to update $NextLevel$;\\
\hspace{7mm} \textbf{endif}\\
\hspace{4.5mm} \textbf{elseif}(($\alpha_{L} \leq P(U(S))\leq P(L(S))$) or \\
\hspace{14.5mm} ($\alpha_{H} \leq P(U(S)) \leq P(L(S))<\alpha_{L}$)) \\
\hspace{7mm}     Handle MSC/NSC to update $NextLevel$;\\
\hspace{4.5mm} \textbf{endif} \\
\hspace{4.5mm}\textbf{endfor} //for each child \\
\hspace{2.7mm}\textbf{endfor} //for each P \\
\textbf{endwhile}\\
$\qreturn$  $\Theta$;
\end{algorithm}
\vspace{-2mm}

Next we describe the IIM algorithm that is used to compute KWII from the attribute sets output by $CIM$.

\subsubsection{IIM Algorithm}
Let $\nu \subseteq \zeta$ be the set of attributes present in $\Theta$ (combinations output by $CIM$). Let $K$ be maximum order of combinations to be explored.
Assuming the sample size to be $N$ and the cardinality of the set of values of the $K^{th}$ order combination to be $V$, $K$ is chosen such that $N/V \geq 5$. The following algorithm computes the $KWII$ of attribute sets of order $\leq K$.
\vspace{-2mm}
\begin{algorithm}{IIM}[\nu,K]
{
 \qinput
  {
   $\nu$ (set of attributes present in $\Theta$),
   $K$ (order of the largest attribute set $\in \Theta$)
  }
  \qoutput
  {
    $\Lambda$(set of combinations and their KWII)
  }
}
$\Lambda \qlet $ entropies of all subsets $\tau$ of $\nu$ s.t. $|\tau|\leq K$;\\
\textbf{for} $A_i \in \nu$ \textbf{do}\\
\hspace{2mm}\textbf{for} each subset $X$ of $\nu/\{A_i\}$, s.t. $|X| < K$ \textbf{do}\\
\hspace{5mm} $\Lambda(X \cup \{A_i\}) \qlet \Lambda(X \cup \{A_i\}) - \Lambda(X)$\\
\hspace{2mm}\textbf{endfor}\\
\textbf{endfor}\\
$\qreturn$ $\Lambda;$
\end{algorithm}
\vspace{-2mm}
\noindent In $IIM$, the array $\Lambda$ is indexed by attribute combinations. We initialize $\Lambda$ with entropies of all subsets of $\nu$ containing upto $K$ attributes (line 1). For example, with 3 attributes $A_1, A_2, A_3$ and $K=2$, $\Lambda(\{A_1\}) = H(A_1)$, $\Lambda(\{A_2\}) = H(A_2)$, $\Lambda(\{A_3\}) = H(A_3)$, $\Lambda(\{A_1;A_2\}) = H(A_1;A_2)$, $\Lambda(\{A_1;A_3\}) = H(A_1;A_3)$ and $\Lambda(\{A_2;A_3\}) = H(A_2;A_3)$.
In the end, $\Lambda$ shall contain negative of KWII values for each attribute combination.

\vspace{-3mm}

\subsubsection{The CIM\textsubscript{CA} Algorithm}
Very few modifications are required to CIM to use it for class associated analysis (the CIM\textsubscript{CA}) algorithm):- (1)Use $CACI(S)$ instead of $TCI(S)$ computations, (2) substitute \textbf{Handle HSC} and \textbf{Handle MSC/NSC} procedures with \textbf{Handle HSC\textsubscript{CA}} and \textbf{Handle MSC\textsubscript{CA}/NSC\textsubscript{CA}} procedures respectively, (3) use COI\textsubscript{CA} and SCOI\textsubscript{CA}, (4) use $S=\{A_i;A_j;C\}$ in line 2, and (4) in lines 27-30 do not call procedure \textbf{Handle MSC\textsubscript{CA}/NSC\textsubscript{CA}} if $U(S) = H(C)$ and $\alpha_{L} \leq P(U(S))\leq P(L(S))$ (using condition 4 from subsection \ref{CACI_pruning}).

\vspace{-3mm}

\subsubsection{Modifications to IIM}
For supervised analysis, the combinations output by CIM\textsubscript{CA} algorithm will constitute the input for IIM. We do not need any changes to IIM described above.
Let $\nu \subseteq \zeta$ be the set of attributes (including $C$) present in $\Theta$ (combinations output by CIM\textsubscript{CA}). Let $K$ be the maximum count of predictor attributes in combinations to be explored. With these inputs, the IIM algorithm can be used unchanged for supervised analysis. Once the set of combinations and their KWII is output by IIM, remove those combinations that do not contain $C$. The remaining combinations shall contain negative of KWII values for combinations containing predictor attributes and class $C$.

\vspace{-3mm}

\section{Experimental Results}
In this section we present the experimental result to highlight the performance of our algorithms. In all our experiments, unless otherwise stated, we have set parameters $\alpha_{High}$, $\alpha_{Low}$ and $\Delta$ to $10^{-8}$, $10^{-2}$ and 0.75 respectively. One can set the $\alpha$'s depending on the experiment and data size, e.g. one would set them conservatively to adjust for multiple comparisons. The $\Delta$ can be set to a value $> 0.7$ depending on how much redundancy one wants to remove from the data. Also, for all experiment, the 10,000 permutations are used for evaluating the significance of each KWII at a significance level 0f 0.0001. We use NIFS and mRMR for comparison purposes. NIFS \cite{NIFS} was run with parameter values $\alpha =$ 0.2 and $\beta =$ 0.8 as used in the paper.

\vspace{-3mm}

\subsection{Unsupervised Analysis}

\vspace{-2mm}

\subsubsection{Experiment 1} Here, we evaluate the effectiveness of our mining methods in detecting attribute associations using a synthetic data set in the absence of a class variable. The data consists of 15 binary attributes and 200 samples and three associations are planted in the data. The associations embedded in the data are (1)$A_1=A_2\oplus A_3$,(2)$A_6=A_7\oplus A_8\oplus A_9$, and (3)$A_{11}=A_{12}\oplus A_{13}\oplus A_{14}$ where $\oplus$ denotes exclusive-or operation. In addition, noise is added by flipping each of $A_1, A_6$ and $A_{11}$ with error probability $p$. We repeat the experiment 100 times. Figure \ref{Fig2}A and B show the TCI and significant KWII mined by CIM and IIM, respectively for $p=0.1$. The significance of each KWII was determined using a pvalue of 0.001. The results are presented graphically as a \emph{spectra} of TCI/KWII values plotted against attribute combinations. Utilizing statistical significance based mining, CIM successfully identifies the embedded associations exactly. The KWII spectra only contains the three strongest associations (pvalue $< 0.001$). Figure \ref{Fig2}C shows that $\%$ of combinations with significant correlation information detected by CIM/IIM and is compared with the two methods NIFS\cite{NIFS} and mRMR\cite{Peng2005}. The error probability $p$ is varied as 0,0.1,0.15,0.2,0.25 and 0.3. Using hard thresholds, NIFS fails to detect the attribute associations when the strength of an associations varies due to noise in the data. CIM/IIM solves this problem by mining with statistical significance levels instead of threshold values. The other method mRMR finds subset of attributes with minimal redundancy among the attributes and a class label attribute. As mRMR requires a class attribute, we have run mRMR separately with: (1) $A_1$, (2) $A_6$, and (3) $A_{11}$ as the class attribute. However, mRMR performs poorly (even at p=0.0) because mRMR uses only mutual information between each attribute and the class to identify the associations.

\begin{figure}[h]
\centering \epsfig{file=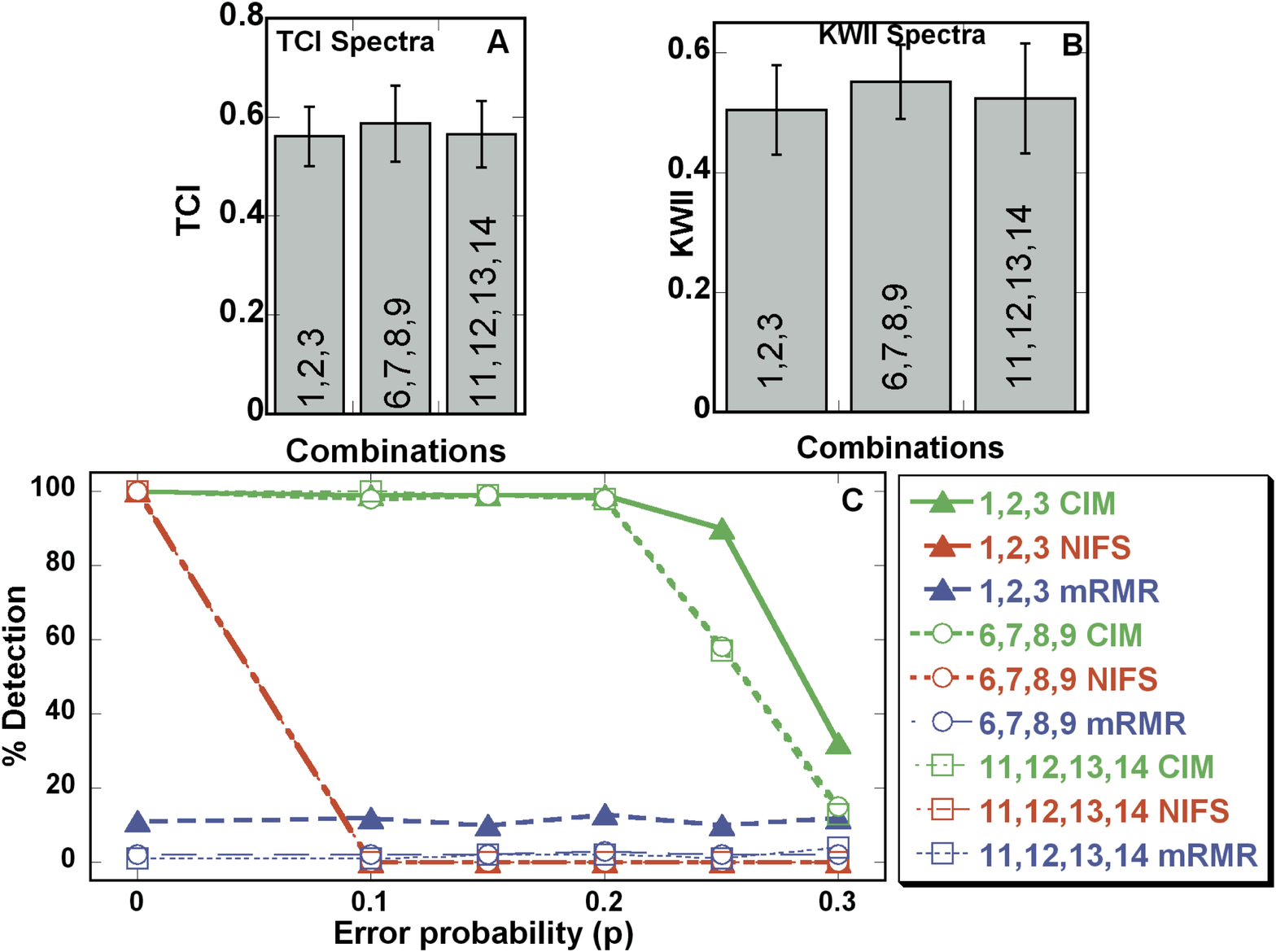,height=2.8in, width=3.4in}
\caption{\label{Fig2}(A) TCI spectra (B) KWII spectra (C) Comparison of CIM/IIM with NIFS and mRMR.}
\end{figure}

\subsubsection{Experiment 2} This experiment is derived from a genetic interaction experiment and mimics the case of pure epistasis \cite{Culvershouse2007} between two SNPs affecting a disease trait. A SNP is a DNA sequence variation in a base pair position at which different sequence alternatives (alleles) exist among individuals in some population. The set of SNPs on a single chromosome of a pair of homologous chromosomes is referred to as a haplotype, and two haplotypes taken together constitutes a genotype. Each SNP usually has two alleles (e.g. $A$, $a$) resulting in three genotype values ($AA$, $Aa$, $aa$). In a case-control experiment, the disease trait is usually binary (0=healthy, 1=diseased). The simulated data in this experiment consists of 16 discrete attributes: $A_1 - A_{15}$ represent SNPs each having 3 genotypic states and $A_{16}$ representing the disease trait is binary. The data consists of 100 samples of $A_{16}$=0 and $A_{16}$=1 each.

\begin{figure}[h]
\centering \epsfig{file=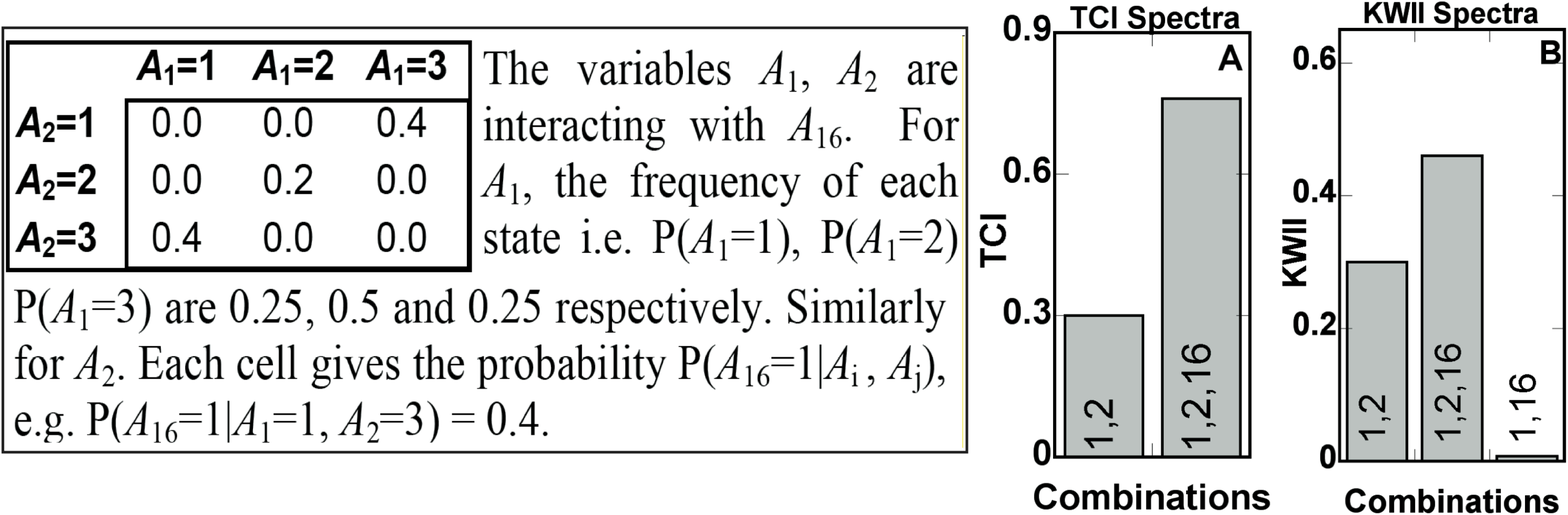,height=1.2in, width=3.5in}
\caption{\label{Fig3}\hspace{-1mm}Association model \& spectra in Experiment 2 for Unsupervised Analysis.}
\end{figure}

Associations are created in the data between $A_1$, $A_2$ and $A_{16}$ using the model in Figure \ref{Fig3}. This results in the following significant combinations:- ($C_1$) $\{A_1;A_2\}$ and ($C_2$) $\{A_1;A_2;A_{16}\}$. Note that $C_1$ is a COI and $C_2$ is a SCOI. Both CIM and IIM are able to identify both the associations (spectra shown in Figure \ref{Fig3}A and B). NIFS identifies only $C_1$ because it assumes that any superset of a set with strong correlation information contains redundant information. However, in this case, $A_{16}$ \emph{can only be identified in combination with $C_1$}, so that $C_2$ contains information about $A_{16}$ not present in $C_1$. Finally, we run mRMR with $A_{16}$ as the class attribute. Combinations $\{A_1;A_{16}\}$ and $\{A_2;A_{16}\}$ have extremely weak mutual information of 0.008 and 0.003 respectively. As mRMR depends on mutual information between each attribute and the class, it fails to identify any combination involving $A_1$ and $A_2$.

\vspace{-3mm}

\subsection{Supervised Analysis}
In this section, we describe the experimental results using our algorithms for supervised analysis (i.e. class attribute is present).

\subsubsection{Experiment 1} In this experiment, the simulated data consists of 15 discrete attributes: $A_1 - A_{15}$ each having 3 states and class attribute $C$. Each of the attributes can be thought to represent genotypes and $C$ represents a binary disease trait. The data consists of 300 samples of $C$=0 and $C$=1 each.
Associations are created in the data between $A_1$, $A_2$ and $C$ using the model in Figure \ref{Fig6}. This results in the following significant combinations:- ($C_1$) $\{A_1;C\}$ and ($C_2$) $\{A_1;A_2;C\}$. This experiment is simulated such that $C_1$ is a COI\textsubscript{CA} and $C_2$ is a SCOI\textsubscript{CA}. Both CIM\textsubscript{CA} and IIM\textsubscript{CA} are able to identify both the associations (Figure \ref{Fig6}) with $100\%$ detection ability and $<5\%$ false combinations in 100 repetitions of the experiment. However, NIFS fail to identify any interaction as it uses hard thresholds. Note that in this case, $A_{2}$ \emph{can only be identified in combination with $C_1$}, so that $C_2$ contains information about $A_{2}$ not present in $C_1$. When we run mRMR with $C$ as the class attribute, it detects
combinations $\{A_1;C\}$ with $100\%$ detection ability. However, because the mutual information $\{A_2;C\}$ is very weak ($\approx 0.003$) mRMR detects $\{A_2;C\}$ with $10\%$ detection ability.

\begin{figure}[h]
\vspace{-3mm}
\centering \epsfig{file=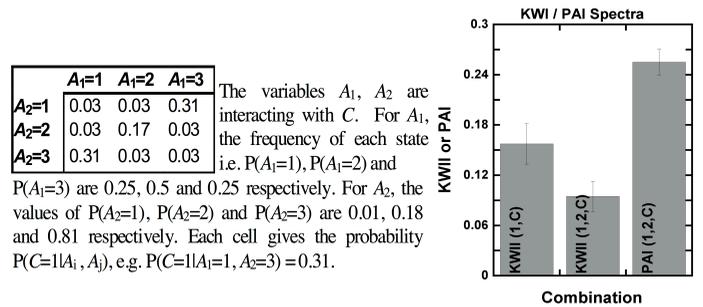,height=1.6in, width=3.6in}
\caption{\label{Fig6}\hspace{-1mm}Association model \& spectra in Experiment 1 for Supervised Analysis.}
\end{figure}

\subsubsection{Experiment 2} The purpose of this experiment is to demonstrate the effectiveness of the redundancy based pruning strategy. In this experiment, the data consists of 15 discrete attributes $A_1$ - $A_{15}$ each with 3 states and a class attribute $C$. A set of complex attribute associations is created involving $A_1$, $A_2$, $A_3$ and $C$. Each of $A_1,A_2,A_3$ represent SNPs having genotypic states $AA/Aa/aa$, $BB/Bb/bb$ and $CC/Cc/cc$ respectively. $C$ stands for a binary disease trait. In addition, redundancy is added by replicating $A_1$ to $A_6$, $A_2$ to $A_7$ and $A_3$ to $A_8$ with $5\%$ error. The data consists of 800 samples of $C$=0 and $C$=1 each. The rule that causes $C$ to be 1 and the CACI and KWII spectra obtained by CIM\textsubscript{CA} and IIM\textsubscript{CA} are shown in Figure \ref{Fig7}. Note that we have effectively removed the redundant attributes ($A_6, A_7, A_8$) and identified all the interacting attributes. Also observe that the KWII spectra complements the CACI spectra by discovering associations like $\{A_1;A_3;C\}$ and $\{A_1;A_2;A_3;C\}$ that are not present in the CACI spectra. Confounded by redundancy, NIFS generates 150 combinations containing attributes $A_1 - A_{10}$ and $C$, but does not contain any combination from the CACI spectra identified by CIM\textsubscript{CA} as their magnitudes are less than 0.8. mRMR is run with $C$ as the class attribute and it identifies attributes $A_1, A_3, C$ in associations $\{A_1;C\}$ and $\{A_3;C\}$ but not $A_2$ because the mutual information $\{A_2;C\}$ is only 0.0008. These show that CIM\textsubscript{CA} and IIM\textsubscript{CA} effectively remove redundancy and are capable of identifying a diverse range of class associated attribute associations.

\begin{figure}[h]
\centering \epsfig{file=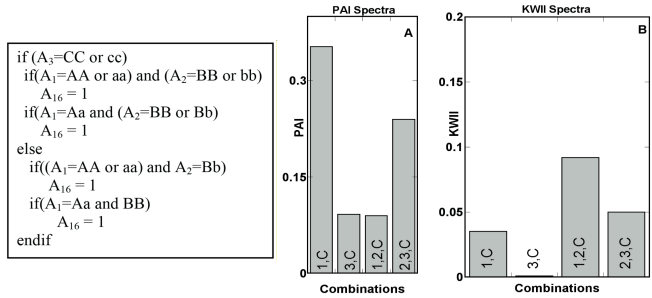,height=1.7in, width=3.6in}
\caption{\label{Fig7}\hspace{-1mm}Association model \& spectra in Experiment 2 for Supervised Analysis.}
\end{figure}

\subsection{Runtime Evaluation}
We have used the following two data sets to evaluate the efficiency of our pruning methods:-(1) Crohn's disease dataset \cite{Daly} is derived from the 616 kilobase
region of human Chromosome 5q31 that may contain a genetic variant responsible for Crohn's disease by genotyping 103 SNPs and contains 144 case and 243 control individuals. (2) Tick-borne encephalitis dataset \cite{Brinza} consists of 58 SNPs genotyped from DNA of 26 patients with severe tick-borne encephalitis virus-induced disease and 65 patients with mild disease. Figure \ref{Fig8} shows the runtime of our mining method (CIM followed by IIM) under the redundancy based and TCI based pruning strategies as well as when both are applied together and none is applied. For both data sets, the missing values were imputed with the most frequent value for that particular SNP. Sample size based pruning is assumed to be active in all the cases. The number of attributes is varied as follows: for each data set, from the set of $N$ attributes, a set of $K$ attributes ($K$ = 10, 20, 30, 40) is randomly selected and removed from the original data. The experiment is repeated 10 times for each data set and the average runtime for each set of $N-K$ attributes is shown. We observe that, the runtime is least when both pruning strategies are active (green, circles). TCI based pruning (blue, squares) achieves better efficiency than redundancy based (red, rhombuses) pruning in both data sets and the runtime increases exponentially when no pruning is applied. These demonstrates the effectiveness of our pruning methods, as the potential search space is exponential in the number of attributes. The results from similar experiment with supervised analysis is shown in Figure \ref{Fig9}.

\begin{figure}[h]
\centering \epsfig{file=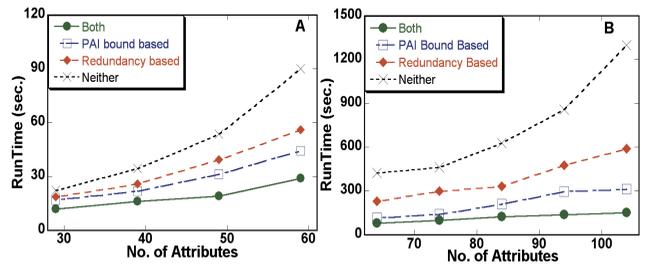,height=1.5in, width=3.5in}
\caption{\label{Fig8}Runtime Evaluation for Unsupervised Analysis with (A) Tick (B) Crohn's Disease.}
\end{figure}

\begin{figure}[h]
\centering \epsfig{file=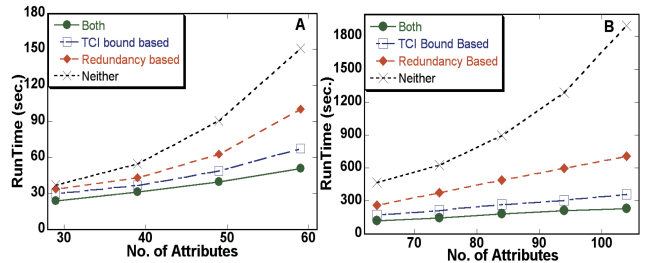,height=1.5in, width=3.5in}
\caption{\label{Fig9}Runtime Evaluation for Supervised Analysis with (A) Tick (B) Crohn's Disease.}
\end{figure}

\subsection{Analysis of Crohn's Disease}
We assess the potential of our mining methods (CIM followed by IIM and CIM\textsubscript{CA} followed by IIM\textsubscript{CA}) for identifying key SNPs involved in the causation of Crohn's disease using data set from Daly et al  \cite{Daly}. The Crohn's disease dataset \cite{Daly} is derived from the 616 kilobase region of human Chromosome 5q31 that may contain a genetic variant responsible for Crohn's disease by genotyping 103 SNPs and contains 144 case and 243 control individuals. The 103 SNPs in the data are numbered 0 to 102. Rioux et al. \cite{Rioux} found 11 SNPs ($IGR2055a\_1$, $IGR2060a\_1$, $IGR2063b\_1$, $IGR2078a\_1$, $IGR2096a\_1$, $IGR2198a\_1$, $IGR2230a\_1$, $IGR2277a\_1$, $IGR2081a\_1$, $IGR3096a\_1$ and $IGR3236a\_1$) with alleles that were associated with risk of Crohn disease. Nine of 11 significant SNPs are present in the data set we analyzed; SNPs $IGR2078a\_1$ and $IGR2277a\_1$ are missing. For association information mining, subjects and SNPs with missing genotypes are eliminated resulting in 40 SNPs with 58 cases and 92 controls. We perform the following two analyses with the data - (1) Mine the association information in the data without the disease phenotype i.e. unsupervised analysis. (2) Mine the association information using our supervised approach with the disease phenotype as the class attribute. In our first analysis, we identify three SNPs $IGR2055a\_1$, $IGR2230a\_1$ and $IGR3236a\_1$ among the combinations with significant $KWII$. In the second analysis where we take the case/control status into account, the five SNPs $IGR2198a\_1$, $IGR2055a\_1$, $IGR3236a\_1$, $IGR2081a\_1$ and $IGR2230a\_1$ are found among the $\{SNP$,$Phenotype\}$ and $\{SNP$,$SNP$,$Phenotype\}$ combinations with significant $KWII$. On closer examination of the data, we found that due to high linkage disequilibrium in the genomic region examined, SNPs $IGR2066a\_1$, $IGR2063b\_1$ and $IGR2096a\_1$ belonged to $Cover(IGR2055a\_1)$ while $IGR3096a\_1$ belonged to $Cover(IGR2230a\_1)$ and were pruned during the redundancy based pruning phase of our mining method. However, each of these SNPs is covered by a representative SNP included in the data, as a result, these SNPs and their associated interactions can be easily recovered using the $Cover$ data structure after $IIM$ completes. For example, consider SNPs $IGR2055a\_1$ and $IGR2066a\_1$. If $IGR2055a\_1$ forms a combination $\{IGR2055a\_1; IGR2198a\_1; C\}$ with significant $KWII$, as we have SNP $IGR2066a\_1$  $\in Cover($SNP $IGR2055a\_1)$, for SNP $IGR2066a\_1$, we can get the combinations with high interaction information as $\{IGR2066a\_1; C\}$ and $\{IGR2066a\_1; IGR2198\_1; C\}$ and test their significance.

\section{Discussion}
In this paper, we have analyzed the problem of mining significant association information between attributes in a data set for both supervised and unsupervised data analysis and have presented novel methods to mine the two types of association information - correlation information and interaction information. Specifically, we have derived the distributional properties of correlation information and bounds on correlation information for the supervised case. We have also developed a novel method for fast permutations to evaluate the significance of interaction information. Using several complex experimental and a real data set, we have critically evaluated the effectiveness and efficiency of our mining strategy. For future work, we would like to explore strategies in making our method scalable for handling large number of attributes as commonly observed in genetic data sets.
\vspace{-2mm}
\bibliographystyle{IEEEtran}


\end{document}